\documentclass[oneside]{amsart}
\usepackage[T1]{fontenc}
\usepackage[latin9]{inputenc}
\usepackage[english]{babel}
\usepackage{verbatim}
\usepackage{float}
\usepackage{textcomp}
\usepackage{amscd}
\usepackage{amsthm}
\usepackage{amstext}
\usepackage{amssymb}
\usepackage{graphicx}
\usepackage{esint}
\usepackage[unicode=true]
 {hyperref}

\makeatletter

\providecommand{\tabularnewline}{\\}

\numberwithin{equation}{section}
\numberwithin{figure}{section}
\theoremstyle{plain}
\newtheorem{thm}{\protect\theoremname}
  \theoremstyle{definition}
  \newtheorem{defn}[thm]{\protect\definitionname}
  \theoremstyle{plain}
  \newtheorem{prop}[thm]{\protect\propositionname}
  \theoremstyle{remark}
  \newtheorem{rem}[thm]{\protect\remarkname}
  \theoremstyle{definition}
  \newtheorem{example}[thm]{\protect\examplename}
  \theoremstyle{plain}
  \newtheorem{lem}[thm]{\protect\lemmaname}
  \theoremstyle{plain}
  \newtheorem{cor}[thm]{\protect\corollaryname}

\def\g{{\mathcal G}}
\def\bR{{\mathbb R}}
\def\R{{\mathbb R}}
\def\N{{\mathbb N}}
\def\Op{\mathrm{Op}}

\makeatother

  \providecommand{\corollaryname}{Corollary}
  \providecommand{\definitionname}{Definition}
  \providecommand{\examplename}{Example}
  \providecommand{\lemmaname}{Lemma}
  \providecommand{\propositionname}{Proposition}
  \providecommand{\remarkname}{Remark}
\providecommand{\theoremname}{Theorem}

\begin{document}
\selectlanguage{english}

\title{quantizations of momentum maps and $G$-systems}
\begin{abstract}
In this note, we give an explicit formula for a family of deformation
quantizations for the momentum map associated with the cotangent lift
of a Lie group action on $\R^{d}$. This family of quantizations is
parametrized by the formal $G$-systems introduced in \cite{BI} and
allows us to obtain classical invariant Hamiltonians that quantize
without anomalies with respect to the quantizations of the action
prescribed by the formal $G$-systems.
\end{abstract}

\author{Benoit Dherin and Igor Mencattini}

\address{Benoit Dherin, Department of Mathematics, University of California,
Berkeley, CA 94720-3840, USA}

\email{dherin@math.berkeley.edu}

\address{Igor Mencattini, ICMC-USP Universidade de Sao Paulo, Avenida Trabalhador
Sao-carlense 400 Centro, CEP: 13566-590, Sao Carlos, SP, Brazil}

\email{igorre@icmc.usp.br}

\maketitle
\tableofcontents{}

\section{Introduction}

The concept of momentum map plays a fundamental role in the classical
description of hamiltonian dynamical systems (in finite and in infinite
dimension), see for example \cite{GSb}. The Marsden-Weinstein reduction
procedure on momentum map level sets (with all of its various generalizations)
is a powerful method to study dynamical systems with symmetries and
to construct new symplectic (Poisson, K\"ahler, hyper-K\"ahler and so
on) manifolds from old ones endowed with a Lie group action preserving
the relevant geometric structures. 

The quantum counterparts of momentum maps (which are special deformation
quantizations, introduced by Ping Xu in \cite{Xu}, of classical momentum
maps regarded as a Poisson maps) and the corresponding reduction procedure
should play a similar fundamental role in the study of quantum systems
with symmetries (see for example \cite{Fedosov}, \cite{Lu}, \cite{Xu},
\cite{La} and also the monograph \cite{LaM}). However, one of the
difficulty in the theory of quantum momentum maps is that explicit
formulas are hard to come by. 

In this paper, we give a such an explicit formula for a family of
deformation quantizations for the momentum map $J$ associated with
the cotangent lift $\tilde{\varphi}$ of an action $\varphi$ of Lie
group $G$ on $\R^{d}$ (Theorem \ref{thm:main}). The result is a
family of deformation quantizations (i.e unital algebra morphisms),
parametrized by the $G$-systems introduced in \cite{BI}, from the
Gutt star-algebra to the standard star-algebra on the cotangent bundle:
\begin{equation}
J^{a}:(C^{\infty}(\mathcal{G}^{*})[[\hbar]],\star_{G})\longrightarrow(C^{\infty}(T^{*}\R^{d})[[\hbar]],\star_{st})\label{eq:Ja}
\end{equation}
where $a$ is a formal $G$-system, that is, a Maurer-Cartan element
in a certain diferential graded algebra of formal amplitudes associated
with the action (see Section \ref{sub:Quantization-of-symmetries}
for a short reminder). 

These quantizations do not satisfy in general the additional conditions
defining quantum momentum maps as described in \cite{Xu} (i.e. that
the star-product on the range of (\ref{eq:Ja}) should be equivariant
with respect to the representation by pullbacks of the cotangent lift
action and that $\frac{i}{\hbar}[J^{a}(v),f]=\tilde{X}^{v}(f)$ must
hold for all $f\in C^{\infty}(T^{*}\R^{d})[[\hbar]]$ and $v\in\mathcal{G}$,
where $\tilde{X}^{v}$ is the fundamental vector field of the cotagent
lift action) but rather deformations of these conditions, controlled
by formal $G$-systems (Proposition \ref{prop:quantum_GSpace} and
Theorem \ref{thm:second}). 

These deformed conditions can be understood in terms of the quantizations
introduced in \cite{BI}: Namely a formal $G$-system $a$ associated
with an action of a Lie group $G$ on $\R^{d}$ produces a representation
of $G$ by formal operators $T_{g}^{a}$ (obtained as asymptotic expansions
of certain Fourier integral operators whose amplitudes are given by
the $G$-system $a$) on the space $C^{\infty}(\R^{d})[[\hbar]]$
of formal functions on $\R^{d}$ (playing the role of the quantum
Hilbert space $L^{2}(\R^{d})$ of states in the formal setting). This
quantization $T^{a}$ lifts to the space $C^{\infty}(T^{*}\R^{d})[[\hbar]]$
of formal functions on the cotangent bundle (playing the role of the
quantum algebra of observables in the formal setting), producing a
representation $\tilde{T}^{a}$ of $G$ on this space that deforms
the representation obtained by pullbacks of the cotangent lift action.

The standard star-product is always equivariant with respect to the
deformed cotangent lift representation $\tilde{T}^{a}$ (Proposition
\ref{prop:quantum_GSpace}). Moreover the deformed condition
\[
\frac{i}{\hbar}[J^{a}(v),f]=\tilde{t}_{v}^{a}f=\tilde{X}^{v}(f)+\mathcal{O}(\hbar),
\]
holds for all $f\in C^{\infty}(T^{*}\R^{d})[[\hbar]]$ and $v\in\mathcal{G}$,
where $\tilde{t}_{v}^{a}$ is the derivative of the lifted representation
$\tilde{T}^{a}$ at the group unit. 

As a by-product, we obtain a family of invariant classical Hamiltonians
$H_{f}^{a}=J^{a}f$, where $f$ in the center of $(C^{\infty}(\mathcal{G}^{*}),\{\,,\,\})$.
These invariant Hamiltonians quantize without anomalies with respect
to the action quantization $T^{a}$ (Theorem \ref{theo:inv}), i.e.,
the quantum Hamiltonians $\hat{H}_{f}^{a}$ are invariant with respect
to the quantized action: $T_{g}^{a}\hat{H}_{f}^{a}T_{g^{-1}}^{a}=\hat{H}_{f}^{a}$,
where $\hat{H}_{f}^{a}$ is the standard quantization of $H_{f}^{a}$
by pseudo-differential operators. $ $

\subsection*{Acknowledgments}

We thank Alberto Cattaneo, Giovanni Felder, Alberto Ibort, Alan Weinstein,
and Maciej Zworski for stimulating discussions, valuable feedback
and insights, as well as the hospitality of the UC Berkeley mathematics
department and the S\~ao Paulo University ICMC, where part of this
project was conducted. We are particularly grateful to Alberto Ibort
for suggesting to us the relation of our construction with quantum
invariant Hamiltonians. B.D. acknowledges partial support from FAPESP
grant 2010/15069-8 and 2010/19365-0.

\section{Setting and results}

In this section, we review how to quantize an action of a Lie group
$G$ on $\R^{d}$ using the $G$-systems introduced in \cite{BI}.
They are Maurer-Cartan elements in a certain differential graded algebra
of amplitudes. Then we recall the notion of quantum momentum maps
as defined in \cite{Xu}, and we conclude with a presentation of our
main results: Namely, Theorem \ref{thm:main} gives a family of deformation
quantization of the momentum map associated with the cotangent lift
of an action on $\R^{d}$ and Proposition \ref{prop:quantum_GSpace}
and Theorem \ref{thm:second} explain how these quantizations satisfy
a deformed version of Ping Xu's original definition.

\subsection{Quantization of symmetries and G-systems\label{sub:Quantization-of-symmetries}}

In \cite{BI}, we introduced a Differential Graded Algebra (DGA) of
amplitudes $(\mathcal{A}_{\varphi}^{\bullet},d,\star)$ associated
with a bounded action $\varphi$ of a Lie group $G$ on $\R^{d}$.
By bounded action, we mean that $\varphi_{g}:\R^{d}\rightarrow\R^{d}$
is a smooth map for each $g\in G$ such that $|\partial^{\alpha}\varphi_{g}^{i}(x)|$
($i=1,\dots,d$) is uniformly bounded for each multi-index $\alpha\in\N^{d}\backslash\{0\}$.
For each positive integer $k$, we set
\[
\mathcal{A}_{\varphi}^{k}:=\Big\{ a:G^{k}\rightarrow S_{2d}(1)\Big\},
\]
where $S_{2d}(1)$ is the space of bounded amplitudes, i.e., families
of smooth functions on $T^{*}\R^{d}$ depending on a parameter $\hbar\in[0,\hbar_{0})$
uniformly bounded on $T^{*}\R^{d}\times[0,\hbar_{0})$ as well as
all their derivatives (see \cite{M} for more details on amplitude
spaces). The differential $d:\mathcal{A}_{\varphi}^{k}\rightarrow\mathcal{A}_{\varphi}^{k+1}$
is given by
\[
(da)_{g_{1},\dots,g_{k+1}}=\sum_{i=1}^{k}(-1)^{i}a_{g_{1},\dots,g_{i}g_{i+1},\dots,g_{k+1}}.
\]

The associative product $\star:\mathcal{A}_{\varphi}^{k}\times\mathcal{A}_{\varphi}^{l}\rightarrow\mathcal{A}_{\varphi}^{k+l}$
is induced from the composition $T^{a}\circ T^{b}=T^{a\star b}$ of
certain Fourier Integral Operators (FIO) acting on $L^{2}(T^{*}\R^{d})$
associated with amplitudes in $\mathcal{A}_{\varphi}^{\bullet}$;
namely, for $a\in\mathcal{A}_{\varphi}^{k}$, we have 
\[
T_{g_{1},\dots,g_{k}}^{a}\psi(x):=\int\psi(\bar{x})a_{g_{1},\dots,g_{k}}(x,\bar{\xi})e^{\frac{i}{\hbar}\langle\bar{\xi},\varphi_{(g_{1}\dots g_{k})^{-1}}(x)-\bar{x}\rangle}\frac{d\bar{x}d\bar{\xi}}{(2\pi\hbar)^{d}}.
\]

\begin{defn}
A \textbf{$G$-system} is a \textbf{Maurer-Cartan element} in $\mathcal{A}_{\varphi}^{\bullet}$,
i.e., an element $a\in\mathcal{A}_{\varphi}^{1}$ satisfying the Maurer-Cartan
equation $da+a\star a=0$. 
\end{defn}

As explained in \cite{BI}, a $G$-system $a$ produces a representation
of $G$ on $L^{2}(T^{*}\R^{d})$. This representation quantizes the
cotangent lift $\tilde{\varphi}$ of the action on $T^{*}\R^{d}$
in the sense of semi-classical analysis (for instance, see symplectomorphism
quantization in \cite{EZ,M}), as we shall see in the next proposition:

\begin{prop}
Let $f$ be a suitably bounded smooth function (e.g. uniformly bounded
by a polynomial in $\xi$) and $a$ be $G$-system associated with
an action $\varphi$ of a Lie group $G$ on $\R^{d}$. Then we have
\begin{equation}
T_{g}^{a}\Op(f)T_{g^{-1}}^{a}=\Op(\tilde{T}_{g}^{a}f),\qquad\tilde{T}_{g}^{a}f=\tilde{\varphi}_{g}^{*}f+\mathcal{O}(\hbar),\label{eq:quant_condition}
\end{equation}
where $\Op(f)$ is the \textbf{standard quantization} of $f$ by pseudo-differential
operators (see \cite{EZ,M} for more details): i.e.,
\begin{equation}
\Op(f)\psi(x):=\int\psi(\bar{x})f(x,\bar{\xi})e^{\frac{i}{\hbar}\langle\bar{\xi},x-\bar{x}\rangle}\frac{d\bar{x}d\bar{\xi}}{(2\pi\hbar)^{d}}.\label{eq:pseudo_diffs}
\end{equation}
\end{prop}

\begin{proof}
A direct computation shows that $T_{g}^{a}\Op(f)T_{g^{-1}}^{a}=\Op(\tilde{T}_{g}^{a}f)$,
where
\[
(\tilde{T}_{g}^{a}f)(x,\xi)=\int a_{g}(x,\bar{\xi})a_{g^{-1}}(\tilde{x},\xi)f(\bar{x},\tilde{\xi})e^{\frac{i}{\hbar}S_{x,\xi}(\bar{x},\bar{\xi},\tilde{x},\tilde{\xi})}\frac{d\bar{x}d\bar{\xi}d\tilde{x}d\tilde{\xi}}{(2\pi\hbar)^{2d}},
\]
with the phase $S_{x,\xi}$ given by
\[
S_{x,\xi}(\bar{x},\bar{\xi},\tilde{x},\tilde{\xi})=\bar{\xi}(\varphi_{g^{-1}}(x)-\bar{x})+\tilde{\xi}(\bar{x}-\tilde{x})+\xi(\varphi_{g}(\tilde{x})-x).
\]
Computing the critical point of $S_{x,\xi}$ w.r.t. the integration
variables, we get
\begin{equation}
\bar{x}=\tilde{x}=\varphi_{g^{-1}}(x),\quad\bar{\xi}=\tilde{\xi}=(T_{x}^{*}\varphi_{g})\xi.
\end{equation}
Using the stationary phase theorem, we get the first term of the asymptotic
expansion of $\tilde{T}_{g}^{a}$; namely,
\begin{eqnarray*}
(\tilde{T}_{g}^{a}f)(x,\xi) & = & a_{g}(x,(T_{x}^{*}\varphi_{g})\xi)a_{g^{-1}}(\varphi_{g^{-1}}(x),\xi)f(\varphi_{g^{-1}}(x),(T_{x}^{*}\varphi_{g})\xi)+\mathcal{O}(\hbar).
\end{eqnarray*}
Now we prove that
\begin{equation}
a_{g}(x,(T_{x}^{*}\varphi_{g})\xi)a_{g^{-1}}(\varphi_{g^{-1}}(x),\xi)=1\label{eq:aa1}
\end{equation}
by doing a similar compution using the relation $T_{g}^{a}T_{g^{-1}}^{a}=\operatorname{id}$,
which holds because $T^{a}$ is a representation as shown in \cite{BI}.
Namely, a direct calculation yields $T_{g}^{a}T_{g^{-1}}^{a}=\Op(g)$,
where 
\begin{equation}
g(x,\xi)=\int a_{g}(x,\bar{\xi})a_{g^{-1}}(\bar{x},\xi)e^{\frac{i}{\hbar}F_{x,\xi}(\bar{x},\bar{\xi})}\frac{d\bar{x}d\bar{\xi}}{(2\pi\hbar)^{d}},\label{eq:g}
\end{equation}
with $F_{x,\xi}(\bar{x},\bar{\xi})=\bar{\xi}(\varphi_{g^{-1}}(x)-\bar{x})+\xi(\varphi_{g}(\bar{x})-x)$.
By injectivity of $\Op$, we have that $g=1$. Using again, as above,
the stationary phase theorem on (\ref{eq:g}), we obtain (\ref{eq:aa1}). 
\end{proof}

Standard quantization defines the \textbf{standard product} of (suitably
bounded, see \cite{EZ,M}) smooth functions on the cotangent bundle:
i.e.,
\begin{equation}
\Op(f\star_{st}g)=\Op(f)\circ\Op(g),\label{eq:standard_product}
\end{equation}

\begin{rem}
\label{rem:formal_version}Observe that both pseudo-differential operators
(\ref{eq:pseudo_diffs}) and the standard product (\ref{eq:standard_product})
can be defined on the space $C^{\infty}(T^{*}\R^{d})[[\hbar]]$ of
formal power series in $\hbar$ with coefficients in the smooth functions
on $T^{*}\R^{d}$ by considering asymptotic expansions of both (\ref{eq:pseudo_diffs})
and (\ref{eq:standard_product}) in the limit $\hbar\rightarrow0$
(again see \cite{EZ,M}). In particular, Equation (\ref{eq:quant_condition})
holds in this formal context, and $\tilde{T}_{g}^{a}$ is now a formal
operator on $C^{\infty}(T^{*}\R^{d})[[\hbar]]$. 
\end{rem}
In this paper, we will be mostly concerned with the formal version
of $\mathcal{A}_{\varphi}^{\bullet}$, also introduced and discussed
in more details in \cite{BI}. Let us briefly outline the construction.
Instead of $\mathcal{A}_{\varphi}^{\bullet}$, we will consider the
DGA $(\mathcal{P}_{\varphi}^{\bullet},d,\star)$ of formal amplitudes.
The construction is similar to the bounded amplitude case, and $\mathcal{P}_{\varphi}^{\bullet}$
can be regarded as the asymptotics of $\mathcal{A}_{\varphi}^{\bullet}$
in the limit $\hbar\rightarrow0$. The formal amplitudes in $\mathcal{P}_{\varphi}^{k}$
are formal power series in $\hbar$ of the form
\[
a_{g_{1},\dots,g_{k}}(x,\xi)=P_{g_{1},\dots,g_{k}}^{0}(x)+\sum_{n\geq1}\hbar^{n}P_{g_{1},\dots,g_{k}}^{n}(x,\xi),
\]
where $P_{g_{1},\dots,g_{k}}^{n}(x,\xi)$ is a polynomial of degree
at most $n$ in $\xi$ with coefficients in the smooth functions on
$\R^{d}$. The corresponding operator $T^{a}$ acts on the space $C^{\infty}(\R^{d})[[\hbar]]$
of formal functions (i.e. formal power series in $\hbar$ with coefficients
in the smooth functions on $\R^{d}$):
\[
T_{g_{1},\dots,g_{k}}^{a}\psi(x)=P^{0}(x)\psi(\varphi_{(g_{1}\cdots g_{k})^{-1}}(x))+\sum_{n\geq1}\hbar^{n}P(x,\frac{1}{i}\partial)\psi_{\Big|\varphi_{(g_{1}\cdots g_{k})^{-1}(x)}}
\]

\begin{defn}
A \textbf{formal $G$-system} is a Maurer-Cartan element in $\mathcal{P}_{\varphi}^{\bullet}$. 
\end{defn}
Similarly to the bounded case, a formal $G$-system $a$ defines a
formal representation $T^{a}$ of $G$ on $C^{\infty}(\R^{d})[[\hbar]]$
deforming the representation by pullbacks,
\begin{equation}
T_{g}^{a}\psi(x)=P_{0}(x)\psi(\varphi_{g^{-1}}(x))+\mathcal{O}(\hbar),\label{eq:formal_rep_on_conf}
\end{equation}
and a corresponding representation $\tilde{T}^{a}$ of $G$ on $C^{\infty}(T^{*}\R^{d})[[\hbar]]$
defined by (\ref{eq:quant_condition}) (see also Remark \ref{rem:formal_version}),
which deforms the representation by cotangent lift pullbacks,
\begin{equation}
\tilde{T}_{g}^{a}f(x)=\tilde{\varphi}_{g^{-1}}^{*}f+\mathcal{O}(\hbar),\label{eq:formal_rep_on_phase}
\end{equation}
where $\tilde{\varphi}_{g}$ is the cotangent lift of $\varphi_{g}$. 
\begin{rem}
There always exists a formal $G$-system; namely, the trivial one:
$a=1$. For this formal $G$-system, the induced representation $T^{a=1}$
on $C^{\infty}(\R^{d})[[\hbar]]$ $ $ is exactly the representation
by pullbacks $\varphi_{g^{-1}}^{*}$ of the action. However, as exemplified
by the Egorov Theorem (see \cite{EZ,M}), one will \textit{not} have
in general that the induced representation $\tilde{T}^{a=1}$ on $C^{\infty}(T^{*}\R^{d})[[\hbar]]$
is the representation by pullbacks $\tilde{\varphi}_{g^{-1}}^{*}$
of the cotangent lift action, but rather, already in the trivial case,
a deformation of it. 
\end{rem}

\subsection{Quantum momentum maps}

In this paragraph, we recall the notion of classical and quantum momentum
maps (we refer the reader to \cite{BFFLS,Burzstyn,Xu} for more details).

Suppose we have an hamiltonian action $\varphi$ of a Lie group $G$
on a symplectic manifold $M$ admitting a momentum map, which, in
general, can be defined as a smooth map $J:M\rightarrow\mathcal{G}^{*}$
such that the hamiltonian vector field with hamiltonian $J^{*}(v)$,
for $v\in\mathcal{G}$ seen as a linear function on $\mathcal{G}^{*}$,
coincides with the fundamental vector field $X^{v}$. Moreover with
require $J$ to be equivariant with respect to symplectic action of
$G$ on the domain and the coadjoint action of $G$ on the range. 

This equivariance implies that $J$ is a Poisson map from $T^{*}\R^{d}$
equipped with the symplectic Poisson bracket to $\mathcal{G}^{*}$
equipped with the Kirillov-Kostant Poisson bracket. The associated
pullback map is thus a Lie algebra morphism
\[
J^{*}:(C^{\infty}(\mathcal{G}^{*}),\{\,,\,\})\longrightarrow(C^{\infty}(M),\{\,,\,\}),
\]
which, restricted to the linear functions on $\mathcal{G}^{*}$, yields
a representation of the Lie algebra $\mathcal{G}$ on the Lie algebra
of classical observables, i.e. $(C^{\infty}(M),\{\,,\,\})$. 

The quantum picture in deformation quantization starts by deforming
the classical Lie algebra in the domain and range of $J^{*}$ into
quantum algebras. Then one deforms $J^{*}$ into a unital algebra
morphism between these quantum algebras, which, similarly to the classical
case, yields a representation of $\mathcal{G}$ into the quantum algebra
of observables quantizing the range.

Let us recall some basic definitions:

\begin{defn}
Let $A$ be a Poisson manifold with Poisson bracket $\{\,,\,\}_{A}$.
A \textbf{deformation quantization} of $A$ is a \textbf{star-product}
$\star_{A}$ on $C^{\infty}(A)[[\epsilon]]$, i.e. a unital associative
product (for which the constant function $1$ is the unit) of the
form
\[
f\star_{A}g=fg+\sum_{n\geq1}\epsilon^{n}B_{n}(f,g),
\]
where the $B_{n}$'s are bidifferential operators such that the quantum
commutator $[\,,\,]_{\star}$ is a deformation of the Poison bracket:
$\frac{1}{\epsilon}[f,g]_{\star}=\{f,g\}+\mathcal{O}(\epsilon)$.
Observe that the formal parameter $\epsilon$ often is taken to be
$\epsilon=\frac{\hbar}{i}$ in concrete example. 
\end{defn}

One natural choice for the quantization of the momentum map domain
is the \textbf{Gutt star-product} $\star_{G}$. It comes from transporting
the associative product on the universal enveloping algebra of $\mathcal{G}$
to the polynomials on $\mathcal{G}^{*}$ via the symmetrization map
(see \cite{Gu}). Another definition of this product is via the asymptotic
expansion of a FIO (see \cite{BA} for instance), this is the definition
we are going to use here: Let $f,g\in C^{\infty}(\mathcal{G}^{*})[[\hbar]]$
, then the Gutt star product $f\star_{G}g$ is the asymptotic expansion
in the limit $\hbar\rightarrow0$ of the integral:
\begin{equation}
(f\star_{G}g)(\theta)=\int f(\theta_{1})g(\theta_{2})e^{\frac{i}{\hbar}\langle\theta,\text{BCH}(v_{1},v_{2})-\langle\theta_{1},v_{1}\rangle-\langle\theta_{2},v_{2}\rangle}\frac{dv_{1}d\theta_{1}dv_{2}d\theta_{2}}{(2\pi\hbar)^{\text{dim}\g}},
\end{equation}
where $\text{BCH}(v_{1},v_{2})$ is the BCH formula. 

One good feature of this star-product is that for two Lie algebra
elements $v,w\in\mathcal{G}$, which we regard as linear functions
on $\mathcal{G}^{*}$, we have
\[
\frac{i}{\hbar}[v,w]_{\star_{G}}=[v,w]_{\mathcal{G}}.
\]
This property allows us to obtain representations of $\mathcal{G}$
into the quantum algebra quantizing $(C^{\infty}(M),\{\,,\,\})$ from
classical momentum map deformation quantizations having for range
the Gutt star-algebra. 

\begin{defn}
Suppose we have two star-products $\star_{A}$ and $\star_{B}$ quantizing
the Poisson algebras $(C^{\infty}(A),\{\,,\,\}_{A})$ and $(C^{\infty}(B),\{\,,\,\}_{B})$,
respectively. A \textbf{deformation quantization }$\hat{\phi}$\textbf{
}from $(C^{\infty}(A)[[\epsilon]],\star_{A})$ to $(C^{\infty}(B)[[\epsilon]],\star_{B})$
\textbf{of a Poisson} map $\phi$ from $B$ to $A$ is a unital algebra
morphism of the form
\[
\hat{\phi}f=\phi^{*}f+\sum_{n\geq1}\epsilon^{n}D_{n}(f),
\]
where the $D_{n}$'s are differential operators. Again, in concrete
examples, one often has $\epsilon=\frac{\hbar}{i}$.
\end{defn}

In \cite{Xu}, Ping Xu introduced the notion of a quantum momentum
map, which is a special deformation quantization of the classical
momentum map regarded as Poisson map. Let us recall his definition,
which involves the notion of quantum $G$-spaces. 

\begin{defn}
\label{def: QQM}(A version of Ping Xu's definition \cite{Xu}). Let
$M$ be a symplectic manifold with a hamiltonian action $\varphi$
of $G$ on $M$ admitting a momentum map $J.$ 

A star-product $\star$ on $C^{\infty}(M)[[\hbar]]$ is \textbf{$G$-equivariant}
if the pullback action $\varphi^{*}$ acts on it by unital algebra
morphisms. The data of a hamiltonian action together with a $G$-equivariant
star-product as above is called a \textbf{quantum $G$-space.} 

A \textbf{quantum momentum map,} quantizing $J$, is a deformation
quantization of $J$ having for domain the Gutt star-algebra and such
that

(1) its range $(C^{\infty}(M)[[\hbar]],\star)$ is a quantum $G$-space,

(2) the condition $\frac{i}{\hbar}[Q(J)(v),f]=X^{v}(f)$ must hold
for all $f\in C^{\infty}(M)[[\hbar]]$. 
\end{defn}

\begin{rem}
The original definition of a quantum momentum map is that of a unital
algebra morphism from the universal enveloping algebra $\mathcal{U}(\mathcal{G}_{\hbar})$
(where $\mathcal{G}_{\hbar}$ is the Lie algebra with Lie bracket
rescaled by a $\hbar$) to a quantum $G$-space $(C^{\infty}(M)[[\hbar]],\star)$
(satisfying also Condition (2)). As already noticed in \cite{Xu},
one can equivalently use the Gutt star-algebra $(C^{\infty}(\mathcal{G}^{*})[[\hbar]],\star_{G})$
as domain for the quantum momentum map. 
\end{rem}

\subsection{Main results}

One difficulty in the theory of quantum momentum maps is that explicit
examples and formulas are hard to come by (except notably for \cite{Hamachi}).
The main result of this note consists in an explicit formula for a
family of deformation quantizations (parametrized by formal $G$-systems)
for the momentum map associated with the cotangent lift of a Lie group
action on $\R^{d}$. Although our quantizations do not satisfy Conditions
(1) and (2) of Ping Xu's original definition (Definition \ref{def: QQM}),
they satisfy deformations of them, controlled by formal $G$-systems.
Let us recall some terminology. 

A smooth action $\varphi$ of a Lie group $G$ on $\R^{d}$ determines,
via cotangent lift, an Hamiltonian action $\tilde{\varphi}$ of $G$
on the cotangent bundle $T^{\ast}\bR^{d}$, which we identify with
$\R^{2d}=\R_{x}^{d}\oplus\R_{\xi}^{d}$ endowed with the canonical
symplectic form $\omega=\sum_{i}d\xi^{i}\wedge dx^{i}$. This action
has a momentum map $J:T^{*}\R^{d}\rightarrow\mathcal{G}^{*}$, where
$\mathcal{G}^{*}$ is the dual of the Lie algebra $\mathcal{G}$ of
the Lie group $G$. It is given by
\begin{equation}
J(x,\xi)=-\langle\xi,X^{\cdot}(x)\rangle\label{eq:cmomap}
\end{equation}
where $X^{\cdot}:\mathcal{G}\rightarrow\operatorname{Vect}(\R^{d})$
is the induced infinitesimal action, i.e., $X^{v}(x)$ is the fundamental
vector field associated with the element $v\in\g$. The sign in (\ref{eq:cmomap})
comes from choosing the symplectic form on the cotangent bundle to
be $\omega$ as above (instead of $-\omega$). 

Let us state here the main theorem, which we will prove later on. 

\begin{thm}
\label{thm:main}Let $a\in\mathcal{P}_{\varphi}^{1}$ be a formal
$G$-system associated with a smooth action $\varphi$ of a Lie group
$G$ on $\R^{d}$. The asymptotic expansion in the limit $\hbar\rightarrow0$
of the map

\begin{equation}
J^{a}(u)(x,\xi)=e^{-\frac{i\langle x,\xi\rangle}{\hbar}}\int_{\g\times\g^{\ast}}^{\operatorname{formal}}u(\theta)a_{\exp(v)}(x,\xi)e^{\frac{i}{\hbar}S_{(x,\xi)}(\langle\theta,v\rangle)}\frac{dvd\theta}{(2\pi\hbar)^{\dim G}}\label{eq:moment}
\end{equation}
where

\begin{equation}
S_{(x,\xi)}(\theta,v)=\langle\xi,\varphi_{\exp(-v)}(x)\rangle-\langle\theta,v\rangle\label{eq:phase}
\end{equation}
is a deformation quantization of the momentum map $J$ above from
$C^{\infty}(\g^{\ast})[[\hbar]],$ endowed with the Gutt star-product
to $C^{\infty}(T^{*}\R^{d})[[\hbar]]$ endowed with the standard star-product.
The integral sign $\int^{\textrm{formal}}$ means that (\ref{eq:moment})
is identified its asymptotic expansion in $\hbar$ in the limit $\hbar\rightarrow0$,
as prescribed by the stationary phase theorem (see remark below). 
\end{thm}

\begin{rem}
\textit{\label{rem:Analytical-meaning}Analytical meaning of (\ref{eq:moment})}.
The phase $S_{x,\xi}(v,\theta)$ in (\ref{eq:moment}) should be actually
understood only as a germ of function, since the exponential map $\exp$
is only defined from a neighborhood of $0$ in the Lie algebra. Therefore
the phase is not defined on whole integration domain $T^{*}\mathcal{G}^{*}$,
but only on a neighborhood of its zero section (which depends on the
germ representative we choose). For (\ref{eq:moment}) to makes sense
as an integral, one needs to throw in the integral a compactly supported
cutoff function $\chi_{(x,\xi)}(\theta,v)$ whose support contains
the critical point $(J(x,\xi),0)$. With this cutoff function, the
integral becomes absolutely convergent, and all the operations permitted
for absolutely convergent integrals will now apply. (This observation
will justify the computations we will perform later on to prove, among
other things, that $J^{a}$ is a unital algebra morphism.) 

The problem at this point is that, if we choose another cutoff function,
the value of the integral will change, since, in fact, we integrate
over a different domain. To remedy this, one should consider the limit
$\hbar\rightarrow0$ after integration, which does not depend on the
choice of the cutoff function, as the stationary phase Theorem guarantees
(see \cite{EZ,M}). We then identify (\ref{eq:moment}) with its asymptotic
expansion in $\hbar$ in the limit $\hbar\rightarrow0$, which is
independent of any cutoff function, leaving the integral always well
defined: This is the meaning of the special integration sign $\int^{\textrm{formal}}$
in (\ref{eq:moment}). 

This remark will apply to all integrals we encounter in this paper.
For the sake of notational brevity, we will avoid to put the cutoff
function each time, and it will be understood that we are dealing
with the asymptotic expansion (a formal power series) of the resulting
absolutely convergent integral. Moreover, again for the sake of notational
simplicity, we will use the standard integral sign $\int$ instead
of the more correct $\int^{\textrm{formal}}$ for most of the integrals
coming in the remaining of this paper. 

Note also that, at times, integral (\ref{eq:moment}) also makes sense
as a non-formal integral (provided one chooses an appropriate space
of functions on $\mathcal{G}^{*}$) as the following example shows
(we will come back to this issue at the end of this paragraph in Remark
(\ref{rem:geom_meaning})):
\end{rem}

\begin{example}
\label{exa:translations}Consider with the action of $\R^{d}$ on
itself by translations and take $a$ to be the trivial $G$-system:
$a=1$. In this case, the quantized action is $T_{v}^{a}\psi(x)=\psi(x-v)$,
its lift to the cotangent bundle is $(\tilde{T}_{v}^{a}f)(x,\xi)=f(x-v,\xi)$,
and the corresponding quantum momentum map is
\[
J^{a}(u)(x,\xi)=\int_{\R^{2d}}u(\bar{\theta})e^{\frac{i}{\hbar}\langle\bar{v},-\xi-\bar{\theta}\rangle}\frac{d\bar{v}d\bar{\theta}}{(2\pi\hbar)^{d}}=u(-\xi).
\]
Here it is easy to see that $J^{a}$ is a unital algebra morphism:
the constant function $1$ on $(\R^{d})^{*}$ is sent to the constant
function $1$ on $T^{*}\R^{d}$; the product of two functions $hk$
on $(\R^{d})^{*}$ (which corresponds to the Gutt star-product when
the Lie group we start with is abelian) is sent to $J^{a}(hk)(x,\xi)=h(-\xi)k(-\xi)$,
which is the standard product of two functions on $T^{*}\R^{d}$ depending
only on $\xi$ (it comes from the asymptotic expansion of $\star_{st}$,
see \cite{EZ,M}). 
\end{example}

The quantizations $J^{a}$ we propose in Theorem \ref{thm:main} do
not satisfy Conditions (1) and (2) of Definition \ref{def: QQM} but
rather deformations of them controlled by $G$-systems. Namely, the
standard product defined by the composition of pseudo-differential
operators as in (\ref{eq:standard_product}) is in general not $G$-equivariant
for cotangent lift actions (unless the action on $\R^{d}$ we start
with is linear). Thus, $(C^{\infty}(T^{*}\R^{d})[[\hbar]],\star_{st})$
is in general not a quantum $G$-space in the sense of Definition
(\ref{def: QQM}), and Condition (1) is not satisfied. However, we
have the following:

\begin{prop}
\label{prop:quantum_GSpace}Let $a$ be a formal $G$-system and let
$\tilde{T}^{a}$ be the induced action on $C^{\infty}(T^{*}\R^{d})[[\hbar]]$
as in (\ref{eq:quant_condition}). The standard product is $G$-equivariant
for this action:
\[
(\tilde{T}_{g}^{a}f)\star_{st}(\tilde{T}_{g}^{a}g)=\tilde{T}_{g}^{a}(f\star_{st}g),
\]
for all $g\in G$ and $f,g\in C^{\infty}(T^{*}\R^{d})[[\hbar]]$. 
\end{prop}

\begin{proof}
Using the definition of $\tilde{T}^{a}$ in (\ref{eq:quant_condition})
and that of the standard product, we have that
\begin{eqnarray*}
\Op((\tilde{T}_{g}^{a}f)\star_{st}(\tilde{T}_{g}^{a}g)) & = & \Op(\tilde{T}_{g}^{a}f)\circ\Op(\tilde{T}_{g}^{a}f)\\
 & = & \tilde{T_{g}^{a}}\circ\Op(f)\circ\tilde{T}_{g^{-1}}^{a}\circ\tilde{T}_{g}^{a}\circ\Op(g)\circ\tilde{T_{g}^{a}}\\
 & = & T_{g}^{a}\Op(f\star_{st}g)T_{g^{-1}}^{a}\\
 & = & \Op(\tilde{T}_{g}^{a}(f\star_{st}g)),
\end{eqnarray*}
which proves our claim, since $\Op$ is injective. 
\end{proof}

Thus $(C^{\infty}(T^{*}\R^{d})[[\hbar]],\star_{st})$ is a kind of
quantum $G$-space but for the deformed action $\tilde{T}^{a}=\tilde{\varphi}^{*}+\mathcal{O}(\hbar)$
(note that corrections in $\hbar$ are present even in the case $a$
is the trivial $G$-system $a=1$). Condition (2) has also a deformed
analog, which we will be able to prove only later on though:

\begin{thm}
\label{thm:second}For all $v\in\g$, seen as a linear function on
$\mathcal{G}^{*}$, we have that
\begin{equation}
J^{a}(v)=J^{\ast}v+\frac{\hbar}{i}(D_{e}a)v,\quad t_{v}^{a}=Op\big(\frac{i}{\hbar}J^{a}(v)\big),\quad\frac{i}{\hbar}[J^{a}(v),f]=\tilde{t}_{v}^{a}(f),\label{eq:deformation_Cond2}
\end{equation}
where $t_{v}^{a}=(D_{e}T^{a})v$, $\tilde{t}_{v}^{a}=(D_{e}\tilde{T}^{a})v$,
and $D_{e}$ is the derivative w.r.t. the group variable evaluated
at the group unit $e$. 
\end{thm}

Since $\tilde{t}_{v}^{a}f=\tilde{X}^{v}f+\mathcal{O}(\hbar)$, where
$\tilde{X}^{v}$ is the fundamental vector field for the cotangent
lift action (associated with $v\in\mathcal{G}$), we have that the
last identity in (\ref{eq:deformation_Cond2}) is a deformation of
Ping Xu's second condition in Definition \ref{def: QQM}. 

\begin{rem}
For the trivial $G$-system $a=1$, we have, by (\ref{eq:deformation_Cond2}),
that $J^{a}$ coincides with the classical momentum map on linear
elements and, thus, that $t_{v}^{a}=X^{v}$, which agrees with the
fact that, in this case, the induced action is the pullback action
$T_{g}^{a}\psi(x)=\psi(\varphi_{g}^{-1}(x))$. However, the induced
action $\tilde{T}^{a}$ on $C^{\infty}(T^{*}\R^{d})[[\hbar]]$ defined
by (\ref{eq:quant_condition}) does not coincide in general with the
action by cotangent lift pullbacks, even if $a=1$. Therefore its
derivative $\tilde{t}^{a}$ is also in this case a deformation of
the action by cotangent lifts of fundamental vector fields. 

There is a case though when $\tilde{T}^{a}$ coincides with the action
by cotangent lift pullbacks: namely, when the action $\varphi$ on
$\R^{d}$ we start with is linear or affine as in Example \ref{exa:translations}.
(One sees this directly from (\ref{eq:quant_condition}), since, in
the linear case, there is no corrections in $\hbar$.) This implies
that $\tilde{t}_{v}^{a}=\tilde{X}^{v}$ and that Ping Xu's second
condition is exactly satisfied. Therefore, for linear or affine actions
and the trivial $G$-system, $J^{a}$ is a quantum momentum map in
the sense of \cite{Xu}. Actually, our formula provides an explicit
formula to \cite[Example 6.5]{Xu}, where the quantum momentum map
was computed only on linear elements. 
\end{rem}

Let us close this section by a remark on the geometrical meaning of
the oscillatory integral (\ref{eq:moment}) defining our quantization
family.

\begin{rem}
\label{rem:geom_meaning}\textit{Geometrical meaning of (\ref{eq:cmomap}).}
As shown in \cite{SMIII}, a Poisson map from $B$ to $A$ integrates
to a symplectic micromorphism from $T^{*}A$ to $T^{*}B$, which is
a special lagrangian submanifold germ of $\overline{T^{*}A}\times T^{*}B$.
These symplectic micromorphisms always posses a global generating
function (see \cite{SMII}), which allows us to quantize them (i.e.
to associate with them formal operators from $C^{*}(A)[[\hbar]]$
to $C^{*}(B)[[\hbar]]$, see \cite{SMIV}) using Fourier integral
operator techniques. Formula (\ref{eq:moment}) can be seen as such
a quantization, where the symplectic micromorphism involved is the
one that integrates the classical momentum map $J$, regarded as a
Poisson map.

When the Poisson map is complete (i.e. when it pulls back complete
hamiltonian vector fields to complete hamiltonian vector fields),
the integrating symplectic micromorphism can be extended to a global
lagrangian submanifold, namely a symplectic comorphism (see \cite{comorphisms}).
In this case, one can expect to obtain bounded operators from some
functional space on $A$ to some other functional space on $B$ as
quantization, instead of their formal asymptotic expansions. Here,
this is reflected by the fact that the exponential map is defined
on the whole Lie algebra only for certain types of Lie groups (e.g.
for the nilpotent ones). For those, the phase becomes a true function
defined on the whole integration domain. It would be interesting (although
analytically challenging) to see if our construction can go beyond
the formal and asymptotic case for nilpotent groups. However, it is
not completely clear to us what are the right functional spaces to
be considered as replacements for the formal spaces $C^{\infty}(\mathcal{G}^{*})[[\hbar]]$
and $C^{\infty}(T^{*}\R^{d})[[\hbar]]$. 
\end{rem}

\section{Proofs of the main results}

In this section, we give proofs of Theorem \ref{thm:main} and Theorem
\ref{thm:second}. One of the main tool will be an asymptotic expansion
of (\ref{eq:moment}) in the limit $\hbar\rightarrow0$, using the
standard Feynman graphical methods. We start off by recalling some
basic facts about the Feynman calculus.

\subsection{Feynman asymptotic expansions}

Let us start with a reminder about asymptotic expansions of oscillatory
integrals in terms of Feynman graphs (we refer the reader to \cite{Feynman}
for more details). Consider the integral

\begin{equation}
I(\hbar)=e^{-\frac{i}{\hbar}S(c)}\int_{\R^{d}}^{\textrm{formal}}g_{1}(z)\cdots g_{n}(z)e^{\frac{i}{\hbar}S(z)}\frac{dz}{(2\pi\hbar)^{\frac{d}{2}}},\label{eq:asexp}
\end{equation}
where $S$ is a smooth function on $\R^{d}$ with a unique non-degenerate
critical point $c$, $g_{1},\dots,g_{n}$ are smooth functions on
$\R^{d}$, and (\ref{eq:asexp}) is to be understood as an asymptotic
expansion in the limit $\hbar\rightarrow0$, yielding a formal power
series in $\hbar$, in the sense of Remark \ref{rem:Analytical-meaning}. 

Feynman's Theorem gives the asymptotic expansion (\ref{eq:asexp})
as a sum over certain graphs: namely,

\begin{equation}
I(\hbar)=\frac{e^{\frac{i\pi}{4}\text{sign}\: B}}{\sqrt{\det\vert B\vert}}\Bigg(\sum_{\Gamma\in G_{3\geq}(n)}\frac{(i\hbar)^{|E_{\Gamma}|-|V_{\Gamma}^{{\rm int}}|}}{\vert\operatorname{Aut}(\Gamma)\vert}F_{\Gamma}(S;g_{1},\dots,g_{n})\Bigg),\label{eq:asymfor}
\end{equation}
where $B=D^{2}S(c)$, $\Gamma$ is a Feynman graph with $n$ external
vertices, $|E_{\Gamma}|$ is its number of edges, $|V_{\Gamma}^{{\rm int}}|$
is its number of internal vertices, $F_{\Gamma}$ is the corresponding
Feynman amplitude, and $|\operatorname{Aut}\Gamma|$ is the number
of symmetries of $\Gamma$. Let us now explain what we mean by Feynman's
graphs and amplitudes.

A \textbf{Feynman graph} is a (non-oriented) graph whose vertex set
$V_{\Gamma}$ is partitioned into two disjoint sets $V_{\Gamma}=V_{\Gamma}^{{\rm {ext}}}\sqcup V_{\Gamma}^{{\rm int}}$:
the set of \textbf{external vertices} $V_{\Gamma}^{{\rm ext}}$, which
we will represent on an imaginary line, as in Table \ref{tab:FeynmanGraphs},
and the set of \textbf{internal vertices} $V_{\Gamma}^{{\rm int}}$,
which we will represent above this imaginary line. Multi-edges and
loops are allowed, but each internal vertex must have valence greater
or equal to $3$. We denote by $G_{3\geq}(n)$ the set of isomorphism
classes of Feynman graphs with $n$ external vertices. We now turn
to Feynman's amplitudes.

Let $S$ and $g_{1},\dots,g_{n}$ be smooth functions on $\R^{d}$
as above. Given a Feynman graph $\Gamma\in G_{3\geq}(n)$ with $k$
internal vertices, the corresponding \textbf{Feynman amplitude }$F_{\Gamma}(S;g_{1},\dots,g_{n})$
is a product of $k$ partial derivatives of $S$ (represented by the
internal vertices) and the partial derivatives of the $g_{i}$'s (represented
by the external vertices) all of which are evaluated at the critical
point $c$ and contracted using the tensor $B^{-1}$ (i.e. the inverse
of the Hessian matrix of $S$ evaluated at $c$). The Feynman graph
records which partial derivatives are involved and how contractions
of these partial derivatives are to be done. We can summarize the
procedure as follows:
\begin{enumerate}
\item Label the two extremities of each edge with a coordinate index in
$\{1,\dots,n\}$
\item An internal vertex with $l$ incoming edges labelled with $i_{1},\dots,i_{l}$
will produce a factor $\frac{\partial^{l}S}{\partial z^{i_{1}}\cdots\partial z^{i_{l}}}(c)$
in the amplitude
\item An external vertex $j\in V_{\Gamma}^{{\bf {int}}}$ with $l$ incoming
edges labelled with $i_{1},\dots,i_{l}$ will produce a factor $\frac{\partial^{l}g_{j}}{\partial z^{i_{1}}\cdots\partial z^{i_{l}}}(c)$
in the amplitude
\item Each edge whose extremities are labelled with, say, $i$ and $j$
(such a labelled edge is also called a \textbf{propagator}) will produce
a factor $(B^{-1})^{ij}$ in the amplitude
\end{enumerate}
The resulting terms should be summed up using the Einstein summation
convention. Here is a table representing a few Feynman graphs and
their amplitudes to illustrate the process:

\begin{table}[H]
\begin{centering}
\begin{tabular}{ccccl}
\includegraphics[scale=0.4]{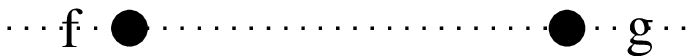} &  & $\rightsquigarrow$ &  & $f(c)g(c)$\tabularnewline
 &  &  &  & \tabularnewline
\includegraphics[scale=0.4]{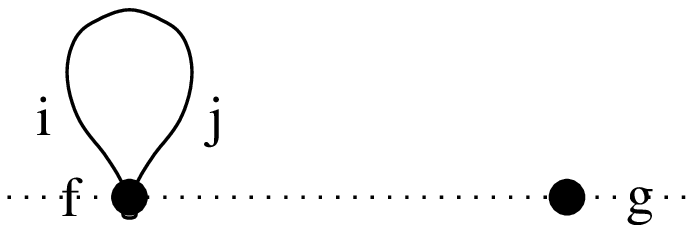} &  & $\rightsquigarrow$ &  & $(B^{-1})^{ij}\frac{\partial^{2}f}{\partial z^{i}\partial z^{j}}(c)g(c)$\tabularnewline
 &  &  &  & \tabularnewline
\includegraphics[scale=0.4]{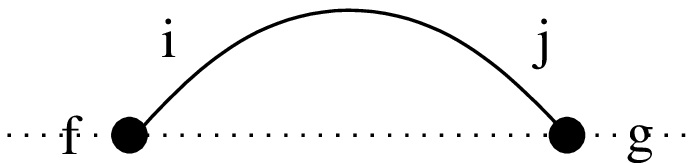} &  & $\rightsquigarrow$ &  & $(B^{-1})^{ij}\frac{\partial f}{\partial z^{i}}(c)\frac{\partial g}{\partial z^{j}}(c)$\tabularnewline
 &  &  &  & \tabularnewline
\includegraphics[scale=0.4]{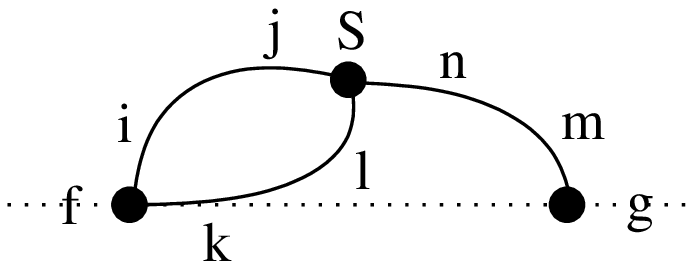} &  & $\rightsquigarrow$ &  & $(B^{-1})^{ij}(B^{-1})^{kl}(B^{-1})^{nm}\frac{\partial^{3}S(c)}{\partial z^{j}\partial z^{l}\partial z^{n}}\frac{\partial f(c)}{\partial z^{i}\partial z^{k}}\frac{\partial g(c)}{\partial z^{m}}$\tabularnewline
 &  &  &  & \tabularnewline
 &  &  &  & \tabularnewline
\end{tabular}
\par\end{centering}

\caption{\label{tab:FeynmanGraphs}On the left some Feynman graphs $\Gamma$
in $G_{3\geq}(2)$ and on the right their corresponding amplitudes
$F_{\Gamma}(S;f,g)$, written using the Einstein summation convention.
We decorated the graphs with labels to make the correspondence more
transparent. }

\end{table}

\subsection{Explicit asymptotic expansion for $J^{a}$}

We now want to use Feynman's theorem (\ref{eq:asymfor}) to obtain
an explicit formula for our family $J^{a}$ of momentum map quantizations
(\ref{eq:moment}) in Theorem \ref{thm:main}. 

The first order of business is to check that the phase in this integral
has a unique non-degenerate critical point and then to compute the
determinant, the signature and the inverse of its Hessian matrix at
this point. Once done, we can use (\ref{eq:asymfor}) in a straightforward
manner. 
\begin{lem}
\label{lem:phase}Consider the phase of integral (\ref{eq:moment}):
\[
S_{x,\xi}(v,\theta)=\langle\xi,\varphi_{\exp(-v)}x\rangle-\langle\theta,v\rangle.
\]
Then, for each $(x,\xi)\in T^{*}\R^{d}$, the phase has a unique critical
point w.r.t. to the $v\theta$-variables; namely,
\[
\Big(v_{0}=0,\quad\theta_{0}=J(x,\xi)\Big),
\]
where $J(x,\xi)=-\langle\xi,X^{v}(x)\rangle$ is the value of the
classical momentum map at $(x,\xi)$. Moreover, at this critical point,
we have that
\[
\det\big(D^{2}S_{x,\xi}(v_{0},\theta_{0})\big)=1\qquad\textrm{and}\qquad\operatorname{sign}(D^{2}S_{x,\xi}(v_{0},\theta_{0}))=0,
\]
and the inverse of the Hessian matrix $B$ is
\[
B^{-1}=\left(\begin{array}{cc}
0 & -\mathbb{I}\\
-\mathbb{I} & -\langle\xi,D^{2}\varphi_{\cdot}(x)\rangle
\end{array}\right).
\]

\end{lem}

\begin{proof}
We get the unique critical point from a direct computation and the
fact that
\[
\frac{d}{dv}\varphi_{\exp(-v)}(x)\vert_{v=0}=-X^{v}(x).
\]
The Hessian matrix of the phase at that point is
\[
B=\left(\begin{array}{cc}
\langle\xi,D^{2}\varphi_{\cdot}(x)\rangle & -\mathbb{I}\\
-\mathbb{I} & 0
\end{array}\right),
\]
from which we get the form of its inverse as well as the fact that
the absolute value of its determinant is always $1$ for all values
$(x,\xi)$. We now prove that the signature of $B$ (i.e. the number
of positive eigenvalues minus the number of negative eigenvalues)
is always equal to zero. First of all, observe that, at $\xi=0$,
the signature of
\[
D^{2}S_{x,\xi}(v_{0},\theta_{0})\vert_{\xi=0}=\left(\begin{array}{cc}
0 & -\mathbb{I}\\
-\mathbb{I} & 0
\end{array}\right),
\]
is zero. Suppose that there exists a $\bar{\xi}$ such that the signature
of $D_{x,\bar{\xi}}^{2}S(v_{0},\theta_{0})\big)$ is non-zero, and
consider the function
\[
f:t\mapsto|\det\big(D_{x,t\bar{\xi}}^{2}S(v_{0},\theta_{0})\big)|,
\]
which is identically equals to $1$ for all values of $t\in[0,1]$.
However, the signature of $D_{x,t\bar{\xi}}^{2}S(v_{0},\theta_{0})$
is zero at $t=0$ and, by assumption, different from zero at $t=1$.
This means that the sign of at least one of the eigenvalues $\lambda_{t}$
of $D_{x,t\bar{\xi}}^{2}S(v_{0},\theta_{0})$ must have changed between
$t=0$ and $t=1$. This implies that there is a value $t_{0}\in[0,1]$
for which $\lambda_{t_{0}}=0$, and, consequently, that $f(t_{0})=0$,
which contradicts the fact that this function is identically $1$.
\end{proof}

The previous Lemma tells us that we can use the Feynman expansion
(\ref{eq:asymfor}) for our quantization family (\ref{eq:moment}),
which yields 
\begin{eqnarray}
(J^{a}u)(x,\xi) & = & \sum_{\Gamma\in G_{3\geq}(2)}\frac{(i\hbar)^{\vert E_{\Gamma}\vert-\vert V_{\Gamma}^{{\rm int}}\vert}}{\vert\operatorname{Aut}\Gamma\vert}F_{\Gamma}(S;u,a)\nonumber \\
 & = & \sum_{\Gamma\in G_{3\geq}(2)}\sum_{l=0}^{\infty}\frac{(-i)^{l}(i\hbar)^{\vert E_{\Gamma}\vert-\vert V_{\Gamma}^{{\rm int}}\vert+l}}{\vert\operatorname{Aut}\Gamma\vert}F_{\Gamma}(S;u,P^{l})\label{eq:exp}
\end{eqnarray}
where $u$ and the $P^{l}$'s are seen as functions of $z=(v,\theta)$
with the particularity that $u$ depends only on $\theta$ and $P^{l}$
depends only on $v$. Because of this particularity and the special
form of $B^{-1}$ many Feynman graphs will have zero amplitude.

We now want to find out how the non-vanishing Feynman graphs look
like. We start with a lemma whose proof can be read off directly from
the form of the phase $S$ and the form of $B^{-1}$ in Lemma \ref{lem:phase}:

\begin{lem}
The only non-vanishing propagators are of two kinds: 

\medskip{}

(1) edges for which one extremity is labelled by $v^{i}$ while the
other extremity is labelled by $\theta_{i}$ ($i=1,\dots,\dim G$).
The corresponding term in the amplitude is $(-1)$. 

\medskip{}

(2) edges for which one extremity is labelled by $\theta_{i}$ while
the other extremity is labelled by $\theta_{j}$ ($i,j=1,\dots,\dim G$).
The corresponding term in the amplitude is $-\sum_{k}\xi_{k}(D_{e}^{2}\varphi^{k})_{ij}$. 

\medskip{}

The only non-vanishing internal vertices of valence $k$ are those
whose incoming edges have labels $v_{i_{1}},\dots,v_{i_{k}}$ (i.e.
none of the labels is taken in the $\theta$-coordinates). 
\end{lem}

In Table \ref{tab:nonVanishingPropagators}, we depict the non-vanishing
propagators and vertices entering in the Feynman expansion of $J^{a}$. 

\begin{table}[H]
\begin{centering}
\begin{tabular}{ccccl}
\includegraphics[scale=0.4]{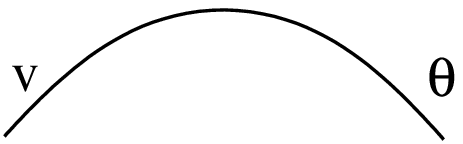} &  & $\rightsquigarrow$ &  & $(B^{-1})^{v\theta}=-1$\tabularnewline
 &  &  &  & \tabularnewline
\includegraphics[scale=0.4]{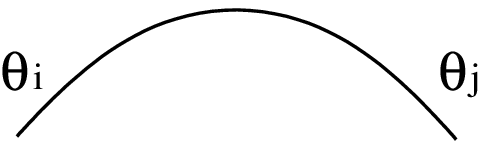} &  & $\rightsquigarrow$ &  & $(B^{-1})^{\theta_{i}\theta_{j}}=-\xi_{k}(D_{e}^{2}\varphi^{k})_{ij}$\tabularnewline
 &  &  &  & \tabularnewline
\includegraphics[scale=0.4]{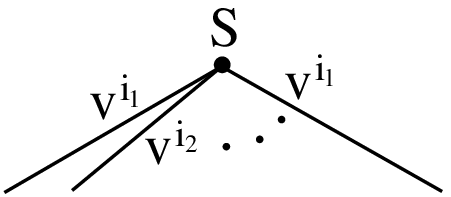} &  & $\rightsquigarrow$ &  & $\frac{\partial^{l}S_{x,\xi}(v_{0},\theta_{0})}{\partial v^{i_{1}}\cdots\partial v^{i_{l}}}=(-1)^{l}\xi_{k}(D_{e}^{l}\varphi^{k})_{i_{1}\dots i_{l}}$\tabularnewline
 &  &  &  & \tabularnewline
\includegraphics[scale=0.4]{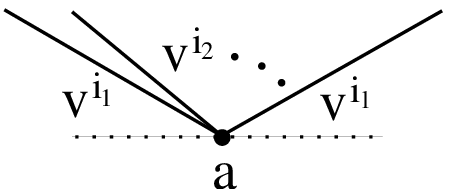} &  & $\rightsquigarrow$ &  & $ $$\frac{}{}$$\frac{\partial^{l}a_{e}(x,\xi)}{\partial v^{i_{1}}\cdots\partial v^{i_{l}}}$\tabularnewline
 &  &  &  & \tabularnewline
\includegraphics[scale=0.4]{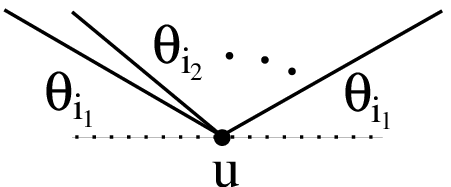} &  & $\rightsquigarrow$ &  & $\frac{\partial^{l}u(J(x,\xi))}{\partial\theta_{i_{1}}\cdots\partial\theta_{i_{l}}}$\tabularnewline
 &  &  &  & \tabularnewline
 &  &  &  & \tabularnewline
\end{tabular}
\par\end{centering}

\caption{\label{tab:nonVanishingPropagators}The non-vanishing propagators
and vertices entering in the Feynman expansion of $J^{a}$.}
\end{table}

\begin{cor}
The only Feynman graphs $\Gamma\in G_{\geq3}(2)$ whose amplitudes
\textup{$F_{\Gamma}(u,a)$ do not vanish for all $u\in C^{\infty}(\g^{\ast})$
and formal }$G$-system $a$ are of the form depicted in Figure \ref{fig:GeneralGraphs}. 
\end{cor}

\begin{figure}[H]
\begin{centering}
\includegraphics[scale=0.4]{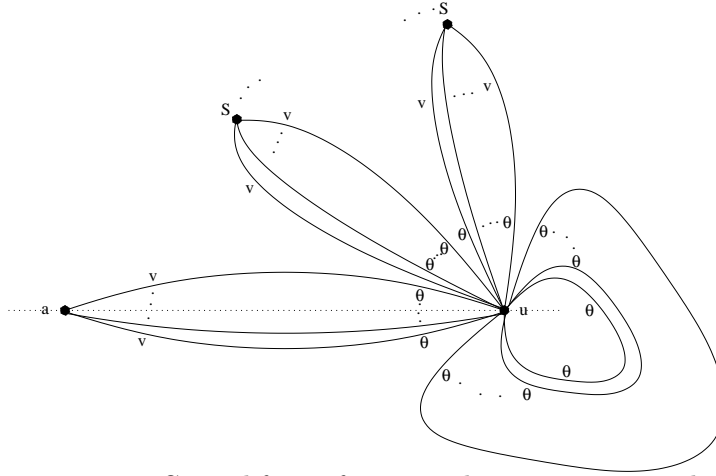}
\par\end{centering}

\caption{\label{fig:GeneralGraphs} General form of non-vanishing Feynman graphs
in the expansion of $J^{a}$.}

\end{figure}

\subsection{Proof of Theorem \ref{thm:main}}

We are now in measure to prove Theorem \ref{thm:main}. We split its
proof into a series of lemmas. The first one shows that $J^{a}$ has
the right first term to be candidate for a deformation quantization
of the classical momentum map:
\begin{lem}
The first term of the expansion (\ref{eq:exp}) is 
\begin{equation}
J^{a}u=J^{\ast}u+\mathcal{O}(\hbar),
\end{equation}
where $J^{*}$ is the pullback of the classical momentum map (\ref{eq:cmomap}). 
\end{lem}

\begin{proof}
The Feynman graph with two external vertices, no internal vertex,
and no edge is the first term of the expansion. The amplitude corresponding
to this term is
\[
F_{\Gamma}(a,u)=u(\theta_{0})a_{\exp(v_{0})}(x,\xi)=(J^{\ast}u)(x,\xi),
\]
since, by Lemma \ref{lem:phase}, the unique critical point is given
by $\theta_{0}=J(x,\xi)$ and $v_{0}=0$, and the formal $G$-system
evaluated at the unit is identically $1$ (i.e. $P_{e}^{0}(x)=1$
and $P_{e}^{n}(x,\xi)=0$ for $n\geq1$). 
\end{proof}

Let us show now the unitality part of Theorem \ref{thm:main}: 

\begin{lem}
We have that $J^{a}(1)=1.$
\end{lem}

\begin{proof}
Using the Feynman expansion (\ref{eq:exp}), we have that the only
graph $\Gamma\in G_{\geq3}(2)$ such that $F_{\Gamma}(a,1)\neq0$
is the one with no edge (any other non-vanishing graph as in Figure
\ref{fig:GeneralGraphs} involves derivatives of the constant function
$1$).
\end{proof}

To complete the proof of Theorem \ref{thm:main}, that is, to show
that $J^{a}$ is an algebra morphism, we need to wait until the next
section. However, we are ready for the proof of Theorem \ref{thm:second}. 

In this section, we will prove that the map
\[
J^{a}:\big(C^{\infty}(\g^{\ast})[[\hbar]],\star_{G}\big)\longrightarrow\big(C^{\infty}(T^{\ast}\bR^{d})[[\hbar]],\ast_{st}\big)
\]
defined in Theorem \ref{thm:main} is an algebra homomorphism, i.e.,
it is a quantization of the classical momentum map of a cotangent
lift action $J$ in (\ref{eq:cmomap}) (we already know from last
section that $J^{a}(1)=1$ and that $J^{a}u=J^{*}u+\mathcal{O}(\hbar)$).
This will complete the proof of Theorem \ref{thm:main}.

At last, we prove that $J^{a}$ is an algebra morphism. Most of the
following computations are formal but can be made rigorous by throwing
in suitable compactly supported cutoff functions in the integrals
as explained in Remark \ref{rem:Analytical-meaning}. For convenience,
we define for $v\in\g$ the following function on $\mathcal{G}^{*}$:
\begin{equation}
\begin{CD}e_{v}(\theta)=e^{\frac{i}{\hbar}\langle v,\theta\rangle},\end{CD}
\end{equation}
whose asymptotic Fourier transform is the translated delta function,
\[
\mathcal{F}_{\hbar}(e_{v})(w)=\delta(w-v),
\]
where the asymptotic Fourier transform of a distribution on $\R^{n}$
is defined by
\[
\mathcal{F}_{\hbar}f(\xi)=\int f(x)e^{-\frac{i}{\hbar}\langle\xi,x\rangle}\frac{dx}{(2\pi\hbar)^{\frac{n}{2}}}.
\]
The following lemma is a standard property of the Gutt star-product,
which we reprove here:

\begin{lem}
\label{lem:1} For all $v,w\in\g$, we have 
\begin{equation}
e_{v}\ast_{G}e_{w}=e_{\text{BCH}(v,uw)}
\end{equation}
\end{lem}
 
\begin{proof}
Let $v_{0},w_{0}\in\g$. Then we have

\begin{eqnarray*}
(e_{v_{0}}\ast_{G}e_{u_{0}})(\theta) & = & \int_{\g\times\g^{*}}\mathcal{F}_{\hbar}(e_{v_{0}})(v)\mathcal{F}_{\hbar}(e_{w_{0}})(w)e^{\frac{i}{\hbar}\langle\theta,BCH(v,w)\rangle}\frac{dvdw}{(2\pi\hbar)^{\dim G}}\\
 & = & \int_{\g\times\g^{*}}\delta(v-v_{0})\delta(u-w_{0})e^{\frac{i}{\hbar}\langle\theta,BCH(v,w)\rangle}\frac{dvdw}{(2\pi\hbar)^{\dim G}}\\
 & = & e^{\frac{i}{\hbar}\langle\theta,BCH(v_{0},w_{0})\rangle}=e_{BCH(v_{0},w_{0})}(\theta).
\end{eqnarray*}
\end{proof}

\begin{lem}
\label{lem:2} For any $v\in\g$, seen as a linear function on $\mathcal{G}^{*}$,
the following identity holds true: 
\begin{equation}
Op\big(J^{a}(e_{v})\big)=T_{\exp(v)}^{a}.
\end{equation}
 \end{lem}
 
\begin{proof}
First we note that
\begin{eqnarray*}
(J^{a}e_{v})(\xi,x) & = & e^{-\frac{i}{\hbar}\langle\xi,x\rangle}\int_{\g}\delta(w-v)a_{\exp(w)}(x,\xi)e^{\frac{i}{\hbar}\langle\xi,\varphi_{\exp(-w)}(x)\rangle}\frac{dw}{(2\pi\hbar)^{\frac{\dim G}{2}}}\\
 & = & e^{-\frac{i}{\hbar}\langle\xi,x\rangle}a_{\exp(v)}(x,\xi)e^{\frac{i}{\hbar}\langle\xi,\varphi_{\exp(-v)}(x)\rangle},
\end{eqnarray*}
which implies that
\begin{eqnarray}
Op\big(J^{a}(e_{v})\big)\psi(x) & = & \int_{\bR^{2d}}\psi(\overline{x})J^{a}(e_{v})(x,\overline{\xi})e^{\frac{i}{\hbar}\langle\overline{\xi},x-\overline{x}\rangle}\frac{d\overline{x}d\overline{\xi}}{(2\pi\hbar)^{d})}\nonumber \\
 & = & \int_{\bR^{2d}}\psi(\overline{x})a_{\exp(v)}(x,\overline{\xi})e^{\frac{i}{\hbar}\langle\overline{\xi},\varphi_{\exp(-v)}(x)-\overline{x}\rangle}\frac{d\overline{x}d\overline{\xi}}{(2\pi\hbar)^{d}}.
\end{eqnarray}
which proves the lemma.
\end{proof}

\begin{lem}
\label{lem:3} For every $v,w\in\g$ we have
\begin{equation}
J^{a}\big(e_{v}\star_{G}e_{w}\big)=J^{a}(e_{v})\star_{st}J^{a}(e_{w}).
\end{equation}
 \end{lem}
 
\begin{proof}
We check this identity at the level of the corresponding pseudo-differential
operators. Then, from Lemma \ref{lem:1} and \ref{lem:2}, we obtain
\[
\Op\big(J^{a}(e_{v}\star_{G}e_{w})\big)=\Op\big(J^{a}(e_{\text{BCH}(v,w)})\big)=T_{\exp\big(\text{BCH}(v,w)\big)}^{a}.
\]
On the other hand
\[
\Op\big(J^{a}(e_{v})\star_{st}J^{a}(e_{w})\big)=\Op\big(J^{a}(e_{v})\big)\circ\Op\big(J^{a}(e_{w})\big)=T_{\exp(v)}^{a}\circ T_{\exp(w)}^{a},
\]
and we can conclude the proof by invoking the injectivity of $\Op$
and the fact that
\[
T_{\exp\big(\text{BCH}(v,w)\big)}^{a}=T_{\exp(v)}^{a}\circ T_{\exp(w)}^{a},
\]
since the operators $\big\{ T_{g}^{a}\big\}_{g\in G}$ define a representation
of the Lie group $G$.
\end{proof}

Then we can conclude the proof of Theorem \ref{thm:main}:
\begin{prop}
The map $J^{a}$ is a an algebra morphism. 
\end{prop}

\begin{proof}
Suppose $\text{dim}\:\g=n$. Let $f,g\in C^{\infty}(g^{\ast})$ and
let us consider their Fourier decomposition: 
\[
f(\theta)=\int_{\mathcal{G}}\mathcal{F}_{\hbar}(v)e_{v}(\theta)\frac{dv}{(2\pi\hbar)^{\frac{n}{2}}}\qquad\text{and}\qquad g(\theta)=\int_{\mathcal{G}}\mathcal{F}_{\hbar}(w)e_{w}(\theta)\frac{dw}{(2\pi\hbar)^{\frac{n}{2}}}.
\]
Then we compute
\begin{eqnarray*}
J^{a}\big(f\star_{G}g\big) & = & J^{a}\left(\Big(\int_{\mathcal{G}}\mathcal{F}_{\hbar}f(v)e_{v}(\theta)\frac{dv}{(2\pi\hbar)^{\frac{n}{2}}}\Big)\star_{G}\Big(\int_{\mathcal{G}}\mathcal{F}_{\hbar}g(w)e_{w}(\theta)\frac{dw}{(2\pi\hbar)^{\frac{n}{2}}}\Big)\right),\\
 & = & \int_{\mathcal{G}\times\mathcal{G}}J^{a}(e_{v}\star_{G}e_{w})\mathcal{F}_{\hbar}f(v)\mathcal{F}_{\hbar}g(w)\frac{dvdw}{(2\pi\hbar)^{n}}\\
 & = & \int_{\g\times\g}\big(J^{a}(e_{v})\star_{st}J^{a}(e_{w})\big)\mathcal{F}_{\hbar}f(v)\mathcal{F}_{\hbar}g(w)\frac{dvdw}{(2\pi\hbar)^{n}}\\
 & = & J^{a}\Big(\int_{\g}\mathcal{F}_{\hbar}f(v)e_{v}\frac{dv}{(2\pi\hbar)^{\frac{n}{2}}}\Big)\star_{st}J^{a}\Big(\int_{\g}\mathcal{F}_{\hbar}g(w)e_{w}\frac{dw}{(2\pi\hbar)^{\frac{n}{2}}}\Big)\\
 & = & J^{a}(f)\star_{st}J^{a}(g),
\end{eqnarray*}
which concludes the proof.
\end{proof}

\subsection{Proof of Theorem \ref{thm:second}}

Let us start with the first identity in Theorem \ref{thm:second}. 

\begin{prop}
For $v\in\g$, seen as a linear function on $\g^{\ast}$, the following
identity holds true
\begin{equation}
J^{a}(v)=J^{\ast}v+\frac{\hbar}{i}(D_{e}a)v.\label{eq:J_on_Lie_elements}
\end{equation}
 \end{prop}
 
\begin{proof}
First we observe that there is no graph $\Gamma$ with internal vertex
such that $B_{\Gamma}(a,v)\neq0$: Suppose that $\Gamma$ has an internal
vertex and that $B_{\Gamma}(a,v)\neq0$. Then one and only one of
the edges stemming out of this vertex must land on the external vertex
labelled with $v$ (if none is landing on $v$, we end up with a zero
propagator, since this internal vertex has (at least) two more edges
decorated by $v$), and if more than one edge is landing on $v$,
then we differentiate twice a linear function. 

Also notice that a Feynman graph such that $B_{\Gamma}(a,v)\neq0$
can not have any self loop based on $v$, because any of these loops
would involve at least a factor of the form $\frac{\partial l_{v}}{\partial v^{i}}$
or a factor of the form $\frac{\partial^{2}l_{v}}{\partial\theta_{i}\partial\theta_{j}}$
in the corresponding amplitude. Since $l_{v}$ is a linear function
depending on $\theta$ only, both factors would yield a zero amplitude.
Loops based on the external vertex decorated by $a$ will also yield
a zero amplitude, since every non-vanishing propagator involve at
least one derivative in the $\theta$-direction and $a$ depends on
the variable $v$ only. 

The only remaining possibilities are the graph $\Gamma_{0}$ with
no internal vertex nor edge and the graph $\Gamma_{1}$ formed by
the two external vertices decorated by $a$ and $v$ respectively
and a single edge with label $v^{i}$ on the $a$ extremity and with
label $\theta_{i}$ on the $v$ extremity (multi-edges would yield
multiple derivations of the linear function $l_{v}$, and hence yield
zero). (We could have seen all this immediately by direct inspection
of Figure \ref{fig:GeneralGraphs}.)

The amplitudes of these two graphs correspond to the two terms in
(\ref{eq:J_on_Lie_elements}).
\end{proof}

Let us prove the second identity of Theorem \ref{thm:second}.

\begin{prop}
\label{prop:Jv}The following identity holds true: 
\begin{equation}
t_{v}^{a}=Op\big(\frac{i}{\hbar}J^{a}(v)\big),\label{eq:derivative_of_T}
\end{equation}
where $v\in\mathcal{G}$ is regarded as a linear function on $\mathcal{G}^{*}$. 
\end{prop}

\begin{proof}
Recall that, by definition, we have
\[
t_{v}^{a}=\big(D_{e}T^{a}\big)(v)=\frac{d}{dt}T_{\exp(tv)}^{a}\vert_{t=0},\qquad v\in\mathcal{G}.
\]
 Then, interchanging derivation and integration and using (\ref{eq:derivative_of_T}),
we obtain
\begin{eqnarray*}
t_{v}^{a} & = & \int_{T^{\ast}\bR^{d}}\psi(\overline{x})\frac{d}{dt}\left(a_{\exp(tv)}(x,\overline{\xi})e^{\frac{i}{\hbar}\langle\bar{\xi},\varphi_{\exp(-tv)}x\rangle}\right)\vert_{t=0}e^{-\frac{i}{\hbar}\langle\overline{\xi},\overline{x}\rangle}\frac{d\overline{x}d\overline{\xi}}{(2\pi\hbar)^{\frac{d}{2}}},\\
 & = & \int_{T^{\ast}\bR^{d}}\psi(\overline{x})\big(D_{e}a(x,\xi\big)(v)-\langle\xi,X^{v}(x)\rangle\big)e^{\frac{i}{\hbar}\langle\bar{\xi},x\rangle}e^{-\frac{i}{\hbar}\langle\overline{\xi},\overline{x}\rangle}\frac{d\overline{x}d\overline{\xi}}{(2\pi\hbar)^{\frac{d}{2}}},\\
 & = & \int_{T^{\ast}\bR^{d}}\psi(\overline{x})\left(\frac{i}{\hbar}J^{a}(v)\right)e^{\frac{i}{\hbar}\langle\bar{\xi},x-\bar{x}\rangle}\frac{d\overline{x}d\overline{\xi}}{(2\pi\hbar)^{\frac{d}{2}}},
\end{eqnarray*}
which concludes the proof. 
\end{proof}

At last, we are ready for the proof of the last identity of Theorem
\ref{thm:second}, which corresponds to a deformation of Ping Xu's
second condition for quantum momentum maps:

\begin{prop}
We have that\textup{
\[
\frac{i}{\hbar}[J^{a}(v),f]=\tilde{t}_{v}^{a}(f)
\]
for all $v\in\mathcal{G}$ and $f\in C^{\infty}(T^{*}\R^{d})[[\hbar]]$. }
\end{prop}

\begin{proof}
Consider Equation (\ref{eq:quant_condition}) evaluated at $g=\exp(tv)$
with $v\in\g$: i.e.
\[
T_{\exp(tv)}^{a}\Op(f)T_{\exp(-tv)}^{a}=\Op(\tilde{T}_{\exp(tv)}^{a}f).
\]
Differentiating this last equation w.r.t. the variable $t$ at $t=0$,
we obtain
\[
\big[t_{v}^{a},Op\big(f\big)\big]=Op(\widetilde{t_{v}^{a}}f\big)
\]
where $\tilde{t}_{v}^{a}=\big(D_{e}\tilde{T}^{a}\big)v$ for $v\in\g$.
Then, by Proposition (\ref{prop:Jv}), we have that
\[
\big[Op\big(\frac{i}{\hbar}J^{a}(v)\big),Op\big(f\big)\big]=Op\big(\widetilde{t_{v}^{a}}f\big)
\]
and finally that
\[
\Op\left(\frac{i}{\hbar}\big[J^{a}(v),f\big]\right)=\Op(\widetilde{t_{v}^{a}}f),
\]
which concludes the proof by injectivity of $\Op$.
\end{proof}

\section{Invariant Hamiltonians}

In this section, we consider certain Hamiltonians invariant w.r.t.
the cotangent lift of an action of a Lie group $G$ on $\bR^{d}$.
In the classical case, invariant Hamiltonians can be obtained as images
of invariant functions in $C^{\infty}(\mathcal{G}^{*})^{G}$ by the
classical comomentum map (i.e. the pullback of $J$ defined in (\ref{eq:cmomap})).
However, the quantum Hamiltonians resulting from the quantization
of these invariant classical Hamiltonians are in general no longer
invariant w.r.t. to the quantized action: anomalies appear. We show
here how to use $G-$systems and their associated quantizations $J^{a}$
defined in (\ref{eq:moment}) to obtain classical invariant Hamiltonians
that are still invariant upon quantization w.r.t. to the quantization
$T^{a}$ (associated with the same $G$-system $a$) of the action.
In other words, we explain how to use $G$-systems to obtain both
quantum symmetries and invariant Hamiltonians with no anomalies upon
quantization. 

Let us start by recalling the classical case.

\subsection{Classical case}

Let $\varphi$ be a smooth action of a Lie group $G$ on $\R^{d}$,
and consider the corresponding hamiltonian action $\tilde{\varphi}$
on the cotangent bundle $T^{*}\R^{d}$ given by cotangent lift with
momentum map $J$ given by (\ref{eq:cmomap}). We will denote by $Ad_{g}^{\sharp}$
the coadjoint action of $G$ on $\g^{\ast}$. It induces an action
on $C^{\infty}(\g^{\ast})$ by pullbacks, which we will still denote
the same way:
\[
(Ad_{g}^{\sharp}f)(\alpha)=f(Ad_{g^{-1}}^{\sharp}\alpha)
\]
for all $f\in C^{\infty}(\g^{\ast})$ and $\alpha\in\g^{\ast}$. Equivariance
of the momentum map implies equivariance of its pullback, the comomentum
map: 
\begin{lem}
For every $f\in C^{\infty}(\g^{\ast})$, the following identity holds
true: 
\[
J^{\ast}\big(Ad_{g}^{\sharp}f\big)=\tilde{\varphi}_{g^{-1}}^{\ast}J^{\ast}(f).
\]

\end{lem}

\begin{defn}
We denote by $C^{\infty}(\mathcal{G}^{*})^{G}$ the space of $Ad^{\sharp}$-invariant
(i.e. $Ad_{g}^{\sharp}f=f$ for all $g\in G$). Observe that this
space coincides with the center $Z$ of the Lie algebra $(C^{\infty}(\mathcal{G}^{*}),\{\,,\,\})$. 
\end{defn}

From the lemma above, we have that

\begin{cor}
\label{cor:invfunctions}Let $f\in C^{\infty}(\g^{\ast})^{G}$ be
an $Ad^{\sharp}$-invariant function. Then, its image by $J^{*}$
is invariant w.r.t. the cotangent lift action:
\begin{equation}
{\tilde{\varphi}_{g}}^{\ast}J^{\ast}f=J^{\ast}f.\label{eq:invf}
\end{equation}
for all $g\in G$
\end{cor}

We will denote the invariant functions on $T^{*}\R^{d}$ as in Corollary
\ref{cor:invfunctions} by
\[
H_{f}:=J^{\ast}f,\qquad f\in C^{\infty}(\mathcal{G}^{*})^{G},
\]
and call them simply \textbf{invariant Hamiltonians}.

\subsection{Quantum case}

Given an invariant Hamiltonian $H_{f}$ as in the last paragraph,
its quantization (here we take the standard quantization for simplicity),
i.e. the pseudo-differential operator
\[
\hat{H}_{f}:=\Op(H_{f})
\]
is in general not a quantum invariant Hamiltonian for the trivial
quantization of the symmetries, since
\begin{equation}
T_{g}\hat{H}_{f}T_{g^{-1}}=\varphi_{g}^{*}\hat{H}_{f}\varphi_{g^{-1}}=\hat{H}_{f}+\mathcal{O}(\hbar),\label{eq:anomalies}
\end{equation}
as exemplified by Egorov's theorem (see \cite{M} for instance). The
corrections in $\hbar$ in (\ref{eq:anomalies}) preventing $\hat{H}_{f}$
to be an invariant quantum hamiltonian (i.e. $T_{g}\hat{H}_{f}T_{g^{-1}}=\hat{H}_{f}$)
are called \textbf{anomalies}. 

Given a formal $G$-system $a$, we deform both an invariant Hamiltonian
$H_{f}$ by using the corresponding quantization of the momentum map
\[
H_{f}^{a}:=J^{a}(f)=H_{f}+\mathcal{O}(\hbar),
\]
and the quantum symmetries by using $T^{a}$ instead of $T$:
\[
T_{g}^{a}AT_{g^{-1}}^{a}=\varphi_{g}^{*}A\varphi_{g^{-1}}^{*}+\mathcal{O}(\hbar),
\]
where $A$ is a pseudo-differential operator. We will show that, in
this case, $H_{f}^{a}$ is still invariant as a classical Hamiltonian
(i.e w.r.t. the cotangent lift action), and, moreover, that the quantization
can be performed without anomalies: i.e.
\[
T_{g}^{a}\hat{H}_{f}^{a}T_{g^{-1}}=\hat{H}_{f}^{a},
\]
where the quantum Hamiltonian is now $\hat{H}_{f}^{a}:=\Op(H_{f}^{a})$.
Let us start by a (rather trivial) example:

\begin{example}
Consider the affine action $\varphi_{v}(x)=x+v$ of $\R^{d}$ on itself
as in Example \ref{exa:translations}. Then the trivial $G$-sysetm
$a=1$ gives $J^{a}f(\xi)=f(-\xi)$, which is obviously invariant
by cotangent lifts of the action. The center $Z$ is the whole space
of functions on $(\R^{d})^{*}$, since the Lie group is abelian, and
the classical invariant Hamiltonians is the space of function $J^{a}(Z)$
on the cotangent bundle depending on the impulsion only. Clearly,
quantization happens in this case without anomalies. However, this
situation is quite degenerate, since, here, $J^{a}$ coincides with
the comomentum map $J^{*}$. Observe that taking $f(\xi)=\xi^{2}$
in the center, $J^{a}f(\xi)=\xi^{2}$ corresponds to the Hamiltonian
of the free particle, whose quantization $H_{f}^{a}$ is the Laplace
operator, which is invariant by quantization of the translations.
Here, the trivial $G$-system gives back the usual story. 
\end{example}

The main result of this paragraph is the following

\begin{thm}
\label{theo:inv} Let $a$ be a formal $G$-system, then
\begin{equation}
T_{g}^{a}\Op\big(J^{a}(f)\big)T_{g^{-1}}^{a}=\Op\big(J^{a}(Ad_{g}^{\sharp}f)\big),\label{eq:invhampre}
\end{equation}
for all $f\in C^{\infty}(\mathcal{G}^{*})$. 
\end{thm}

This result has the following consequence:

\begin{cor}
Given an $Ad^{\sharp}$-invariant function $f\in C^{\infty}(\mathcal{G}^{*})^{G}$,
then the corresponding quantum Hamiltonian $\hat{H}_{f}^{a}=\Op\big(J^{a}(f)\big)$
is invariant w.r.t. to the quantum symmetries $T^{a}$, i.e.
\begin{equation}
T_{g}^{a}\hat{H}_{f}T_{g^{-1}}^{a}=\hat{H}_{f}\label{eq:invham}
\end{equation}
for all $g\in G$.
\end{cor}

\begin{rem}
From (\ref{eq:invham}) and (\ref{eq:quant_condition}), we observe
that $H_{f}^{a}$ is also invariant as a classical Hamiltonian. In
fact:
\[
\Op\big(J^{a}(f)\big)=T_{g}^{a}\Op\big(J^{a}(f)\big)T_{g^{-1}}^{a}=\Op\big(\tilde{\varphi}_{g^{-1}}^{\ast}J^{a}(f)+\mathcal{O}(\hbar)\big),
\]
for all $f\in C^{\infty}(\mathcal{G}^{*})$ and $g\in G$, from which
we get $\tilde{\varphi}_{g^{-1}}^{\ast}J^{a}(f)=J^{a}(f)$ for all
$g\in G$.
\end{rem}

The remaining of this section is devoted to the proof of Theorem \ref{theo:inv}.

\subsection{Proof of the Theorem \ref{theo:inv}}

Throughout this section, we suppose that $\dim\mathcal{G}=n$. We
start with the following

\begin{lem}
\label{lem:four_trans}Let $f\in C^{\infty}(\g^{\ast})$. Then, for
all $g\in G$ 
\begin{equation}
\mathcal{F}_{\hbar}(Ad_{g}^{\sharp}f)(v)=\vert\det Ad_{g}^{\sharp}\vert\mathcal{F}_{\hbar}(f)\big(Ad_{g^{-1}}(v)\big)\label{eq:fide}
\end{equation}
where $\mathcal{F}_{\hbar}$ denotes the asymptotic Fourier transform.
\end{lem}

\begin{proof}
This is a direct computation: 
\begin{eqnarray*}
\mathcal{F}_{\hbar}(Ad_{g}^{\sharp}f)(v) & = & \int_{\g^{\ast}}\big(Ad_{g}^{\sharp}f\big)(\theta)e^{-\frac{i}{\hbar}\langle\theta,v\rangle}\frac{d\theta}{(2\pi\hbar)^{\frac{n}{2}}}\\
 & = & \int_{\g^{\ast}}f(Ad_{g^{-1}}^{\sharp}\theta)e^{-\frac{i}{\hbar}\langle\theta,v\rangle}\frac{d\theta}{(2\pi\hbar)^{\frac{n}{2}}}\\
 & = & \int_{\g^{\ast}}f(\tilde{\theta})\vert\det\: Ad_{g}^{\sharp}\vert e^{-\frac{i}{\hbar}\langle\tilde{\theta},Ad_{g^{-1}}v\rangle}\frac{d\tilde{\theta}}{(2\pi\hbar)^{\frac{n}{2}}}\\
 & = & \vert\det\: Ad_{g}^{\sharp}\vert\mathcal{F}_{\hbar}(f)(Ad_{g^{-1}}v)
\end{eqnarray*}
where $\tilde{\theta}=Ad_{g^{-1}}\theta$ and where we used $\langle Ad_{g}^{\sharp}\tilde{\theta},v\rangle=\langle\tilde{\theta},Ad{g^{-1}}v\rangle$
for all $g\in G$, $\tilde{\theta}\in\g^{\ast}$ and $v\in\g$.
\end{proof}

Moreover

\begin{lem}
\label{lem:side}Given a $G$-system $a$, and for all $f\in C^{\infty}(\g^{\ast})$
and $g\in G$, the following formula holds true 
\[
J^{a}(Ad_{g}^{\sharp}f)(x,\xi)=e^{-\frac{i}{\hbar}\langle\xi,x\rangle}\int_{\g}\mathcal{F}_{\hbar}f(v)a_{g\exp(v)g^{-1}}(x,\xi)e^{\frac{i}{\hbar}\langle\xi,\varphi_{g\exp(-v)g^{-1}}(x)\rangle}\frac{dv}{(2\pi\hbar)^{\frac{n}{2}}}.
\]
 \end{lem}
 
\begin{proof}
Also in this case the proof of the statement follows from a direct
computation. Using Lemma \ref{lem:four_trans} to compute $I=e^{+\frac{i}{\hbar}\langle\xi,x\rangle}J^{a}(Ad_{g}^{\sharp}f)(x,\xi)$,
we obtain 
\begin{eqnarray*}
I & = & \int_{\g}\mathcal{F}_{\hbar}(Ad_{g}^{\sharp}f)(v)a_{\exp(v)}(x,\xi)e^{\frac{i}{\hbar}\langle\xi,\varphi_{\exp(-v)}(x)\rangle}\frac{dv}{(2\pi\hbar)^{\frac{n}{2}}}\\
 & = & \int_{\g}\vert\det\: Ad_{g}^{\sharp}\vert\mathcal{F}_{\hbar}f(Ad_{g^{-1}}v)a_{\exp(v)}(x,\xi)e^{\frac{i}{\hbar}\langle\xi,\varphi_{\exp(-v)}(x)\rangle}\frac{dv}{(2\pi\hbar)^{\frac{n}{2}}}\\
 & = & \int_{\g}\vert\det\: Ad_{g}^{\sharp}\vert\vert\det\: Ad_{g}\vert\mathcal{F}_{\hbar}f(\tilde{v})a_{g\exp(\tilde{v})g^{-1}}(x,\xi)e^{\frac{i}{\hbar}\langle\xi,\varphi_{g\exp(-v)g^{-1}}(x)\rangle}\frac{d\tilde{v}}{(2\pi\hbar)^{\frac{n}{2}}},
\end{eqnarray*}
where we used the change of variable $\tilde{v}=Ad_{g^{-1}}(v)$ and
the identities
\[
\exp(Ad_{g}(v))=g\exp(v)g^{-1}\quad\textrm{and}\quad\vert\det\: Ad_{g}^{\sharp}\vert\vert\det\: Ad_{g}\vert=1,
\]
which hold for every $v\in\g$ and $g\in G$.
\end{proof}

We can now conclude the proof of Theorem \ref{theo:inv}. Using Lemma
\ref{lem:side} to compute $K=\Op\big(J^{a}(Ad_{g}^{\sharp}f)\big)$,
we obtain
\begin{eqnarray*}
(K\psi)(x) & = & \int_{\bR^{2d}}\psi(\overline{x})J^{a}(Ad_{g}^{\sharp}f)(x,\overline{\xi})e^{\frac{i}{\hbar}\langle\overline{\xi},(x-\overline{x})\rangle}\frac{d\overline{x}d\overline{\xi}}{(2\pi\hbar)^{\frac{d}{2}}}\\
 & = & \int_{\g}\int_{\bR^{2d}}\psi(\overline{x})\mathcal{F}_{\hbar}f(v)a_{g\exp(-v)g^{-1}}(x,\overline{\xi})e^{\frac{i}{\hbar}\langle\overline{\xi},\varphi_{g\exp(-v)g^{-1}}(x)-\bar{x}\rangle}\frac{dvd\overline{x}d\overline{\xi}}{(2\pi\hbar)^{\frac{d+n}{2}}}\\
 & = & \int_{\g}\mathcal{F}_{\hbar}f(v)\big(T_{g\exp(v)g^{-1}}^{a}\psi\big)(x)\frac{dv}{(2\pi\hbar)^{\frac{n}{2}}}.
\end{eqnarray*}
Since $a$ is a formal $G$-system, we have
\[
\big(T_{g\exp(v)g^{-1}}^{a}\psi\big)(x)=\big(T_{g}^{a}T_{\exp(v)}^{a}T_{g^{-1}}^{a}\psi\big)(x)
\]
for all $g\in G$, $v\in\g$ and $\psi\in C^{\infty}(\bR^{d})$, and
thus that $K=T_{g}^{a}LT_{g^{-1}}$, where $L$ is the operator defined
by
\[
L=\int_{\g}\mathcal{F}_{\hbar}f(v)T_{\exp(v)}^{a}\frac{dv}{(2\pi\hbar)^{\frac{n}{2}}}.
\]
Let us now compute its action on functions:
\begin{eqnarray*}
(L\psi)(x) & = & \int_{\g}\mathcal{F}_{\hbar}f(v)T_{\exp(v)}^{a}\psi(x)\frac{dv}{(2\pi\hbar)^{\frac{n}{2}}}\\
 & = & \int_{\g}\mathcal{F}_{\hbar}f(v)\bigg[\int_{\bR^{2d}}\psi(\overline{x})a_{\exp(v)}(x,\overline{\xi})e^{\frac{i}{\hbar}\langle\overline{\xi},\varphi_{\exp(-v)}(x)-\overline{x}\rangle}\frac{d\overline{x}d\overline{\xi}}{(2\pi\hbar)^{\frac{d}{2}}}\bigg]\frac{dv}{(2\pi\hbar)^{\frac{n}{2}}}
\end{eqnarray*}
Unwrapping the Fourier transform in this last expression, we obtain
\[
\int_{\bR^{2d}}\psi(\overline{x})\bigg[e^{-\frac{i}{\hbar}\langle\overline{\xi},x\rangle}\int_{\g}\mathcal{F}_{\hbar}(v)a_{\exp(v)}(x,\overline{\xi})e^{\frac{i}{\hbar}\langle\overline{\xi},\varphi_{\exp(-v)}(x)\rangle}\frac{dv}{(2\pi\hbar)^{\frac{n}{2}}}\bigg]e^{\frac{i}{\hbar}\langle\overline{\xi},(x-\overline{x})\rangle}\frac{d\overline{x}d\overline{\xi}}{(2\pi\hbar)^{\frac{d}{2}}},
\]
which we recognize to be $\Op\big(J^{a}(f)\big)\psi(x)$. Thus $K=T_{g}^{a}LT_{g^{-1}}$
is exactly what we wanted to prove: namely,
\[
\Op\big(J^{a}(Ad_{g}^{\sharp}f)\big)=T_{g}^{a}\Op\big(J^{a}(f)\big)T_{g^{-1}}^{a}.
\]

\end{document}